\documentclass[conference]{IEEEtran}
\ifCLASSINFOpdf
\else
\fi
\usepackage{hyperref}
\usepackage{url}

\usepackage{color}
\usepackage{soul}
\usepackage{algorithm}
\usepackage{algpseudocode}
\usepackage{makecell}

\usepackage{amsthm}
\newtheorem{proposition}{Proposition}

\newtheorem{theorem}{Theorem}
\usepackage{algorithm}
\usepackage{algpseudocode}

\usepackage{times}
\usepackage{breqn}
\algrenewcommand\algorithmicrequire{\textbf{Input:}}
\algrenewcommand\algorithmicensure{\textbf{Output:}}

\usepackage{epsfig}
\usepackage{epstopdf}
\usepackage{subfigure}
\usepackage{multirow}
\usepackage{comment}
\usepackage{listings}
\lstdefinelanguage{chapel}
  {
    morekeywords=[1]{
      align, atomic,
      begin, bool, break, by,
      class, cobegin, coforall, complex, config, const, continue,
      delete, dmapped, do, domain,
      else, enum, extern, export,
      for, forall,
      if, imag, in, index, inline, inout, iter,
      label, lambda, let, local, locale,
      module,
      new, nil, noinit,
      on, opaque, otherwise, out,
      param, proc,
      range, real, record, reduce, ref, return,
      scan, select, serial, single, sparse, string, subdomain, sync,
      then, type,
      union, use,
      when, where, while, with,
      yield,
      zip
    },
    morekeywords=[2]{
        int,
        owned,
        real,ref,
        true,
        uint,
        var
    },
    sensitive=true,
    mathescape=true,
    morecomment=[l]{//},
    morecomment=[s]{/*}{*/},
    morestring=[b]",
    literate=*
    {(}{{\textcolor{blue}{(}}}{1}
    {)}{{\textcolor{blue}{)}}}{1}
    {\{}{{\textcolor{blue}{\{}}}{1}
    {\}}{{\textcolor{blue}{\}}}}{1} 
    {()}{{\textcolor{blue}{()}}}{2}
}
\lstdefinestyle{mystyle}{
    backgroundcolor=\color{backcolour},   
    commentstyle=\color{codegreen},
    keywordstyle=[1]\color{magenta},
    keywordstyle=[2]\color{blue},
    numberstyle=\tiny\color{codegray},
    stringstyle=\color{codepurple},
    basicstyle=\ttfamily\footnotesize,
    breakatwhitespace=false,
    captionpos=b,
    keepspaces=true,
    numbers=left,
    numbersep=6pt,
    showspaces=false,
    showstringspaces=false,
    showtabs=false,
    tabsize=2,
    frame=single,
    frameround=tttt,
    xleftmargin=3.4pt,
    xrightmargin=3.4pt
}
\usepackage{xcolor}
\definecolor{codegreen}{rgb}{0,0.6,0}
\definecolor{codegray}{rgb}{0.5,0.5,0.5}
\definecolor{codepurple}{rgb}{0.58,0,0.82}
\definecolor{backcolour}{rgb}{0.95,0.95,0.92}

\usepackage{soul}
\usepackage{xspace}
\usepackage{url}
\usepackage{comment}
\usepackage{bm}

\usepackage{cite}
\usepackage{amsmath,amssymb,amsfonts}
\usepackage{graphicx}
\usepackage{textcomp}
\usepackage{xcolor}
\def\BibTeX{{\rm B\kern-.05em{\sc i\kern-.025em b}\kern-.08em
    T\kern-.1667em\lower.7ex\hbox{E}\kern-.125emX}}

\graphicspath{{figures/}}

\makeatletter
\newcommand{\linebreakand}{%
  \end{@IEEEauthorhalign}
  \hfill\mbox{}\par
  \mbox{}\hfill\begin{@IEEEauthorhalign}
}
\makeatother

\begin{document}
%
\title{HiPerMotif: Novel Parallel Subgraph Isomorphism in Large-Scale Property Graphs}

\author{
    \author{
        \IEEEauthorblockN{
            Mohammad Dindoost,
            Oliver Alvarado Rodriguez, Bartosz Bryg,
            Ioannis Koutis, and
            David A. Bader
        }
        \IEEEauthorblockA{\textit{Department of Computer Science} \\
        \textit{New Jersey Institute of Technology}\\
                Newark, NJ, USA \\
                \texttt{\{md724,oaa9,bb474,ikoutis,bader\}@njit.edu}
        }
    }
}


%


\maketitle

\begin{abstract}
Subgraph isomorphism, essential for pattern detection in large-scale graphs, faces scalability challenges in attribute-rich property graphs used in neuroscience, systems biology, and social network analysis. Traditional algorithms explore search spaces vertex-by-vertex from empty mappings, leading to extensive early-stage exploration with limited pruning opportunities. We introduce HiPerMotif, a novel hybrid parallel algorithm that fundamentally shifts the search initialization strategy. After structurally reordering the pattern graph to prioritize high-degree vertices, HiPerMotif systematically identifies all possible mappings for the first edge (vertices 0,1) in the target graph, validates these edge candidates using efficient vertex and edge validators, and injects the validated partial mappings as states at depth 2. The algorithm then continues with traditional vertex-by-vertex exploration from these pre-validated starting points, effectively pruning the expensive early search tree branches while enabling natural parallelization over edge candidates. Our contributions include the edge-centric initialization paradigm with state injection, a structural reordering strategy achieving up to 5× speedup, rapid edge and vertex validators for attribute-rich graphs, and efficient parallel enumeration over target graph edges. Implemented in the open-source Arachne framework, HiPerMotif achieves up to 66× speedup over state-of-the-art baselines (VF2-PS, VF3P, Glasgow) on diverse datasets where baselines successfully complete execution. Additionally, HiPerMotif successfully processes massive datasets such as the H01 connectome with 147 million edges, which existing methods cannot handle due to memory constraints. Comprehensive evaluation across synthetic and real-world graphs demonstrates HiPerMotif's scalability, enabling advanced analysis in computational neuroscience and beyond.
\end{abstract}

\section{Introduction}

Subgraph isomorphism and monomorphism, critical for detecting patterns in large-scale graphs, underpin discoveries in neuroscience~\cite{sporns2018graph,van2013wu}, systems biology~\cite{mangan2003structure,barabasi2011network}, social network analysis~\cite{ugander2011anatomy,fire2016organization}, and cybersecurity~\cite{akoglu2015graph,noble2003graph}. Recurring substructures, or motifs, reveal network organization and function through tasks like motif detection and network similarity analysis~\cite{alon2007network,lauri2011subgraphs}. Analyzing massive graphs with millions of edges demands scalable, expressive tools. Traditional adjacency-based representations fail to capture domain-specific semantics (e.g., neuron types, protein interactions), necessitating property graphs~\cite{angles2018property,rodriguez2010property} that support multiple vertex and edge properties, significantly increasing computational complexity.

Traditional subgraph matching algorithms follow a vertex-centric approach, starting from empty mappings and incrementally building partial solutions vertex-by-vertex through recursive backtracking. This strategy requires extensive exploration of early search tree levels, where pruning opportunities are limited and parallel efficiency is constrained by sequential vertex ordering decisions. Existing algorithms struggle with five fundamental challenges: (1) inefficient candidate generation producing large search spaces; (2) rigid vertex-ordering heuristics that overlook graph structure; (3) high memory overhead in tracking candidate sets; (4) limited parallelization opportunities in early tree-search stages; and (5) unnecessary exhaustive enumeration despite applications often needing partial results. These limitations become critical in applications like connectome motif detection, where graphs exceed 100 million edges.

We propose HiPerMotif, a hybrid algorithm that fundamentally changes the search initialization strategy while preserving the benefits of established tree-search methods. After structurally reordering the pattern graph to prioritize high-degree vertices, HiPerMotif systematically identifies all possible mappings for the first edge (vertices 0,1) in the target graph. Instead of exploring vertex assignments sequentially from empty states, the algorithm validates these edge candidates using efficient vertex and edge validators, then injects the validated partial mappings as states at depth 2 into traditional vertex-by-vertex search. This approach eliminates the expensive early search tree exploration (depths 0-1) while enabling natural parallelization over edge candidates rather than search tree branches. Implemented in the open-source Arachne framework, HiPerMotif achieves up to 66$\times$ speedup over state-of-the-art baselines (VF2-PS, VF3P, Glasgow) and successfully processes massive datasets, such as the H01 connectome~\cite{shapson2024petavoxel}, that cause existing methods to fail due to memory constraints. The main contributions of this paper are:
\begin{enumerate}
    \item A hybrid algorithmic paradigm that systematically validates all edge mappings for the first edge and injects partial states at depth 2, eliminating costly early search tree exploration while preserving traditional vertex-by-vertex efficiency.
    \item A structural reordering strategy that prioritizes high-degree vertices to identify the optimal starting edge, achieving up to 5$\times$ speedup in search initiation.
    \item Efficient vertex and edge validators that rapidly prune infeasible edge candidates before state injection, reducing computational overhead in attribute-rich property graphs.
    \item A parallel edge enumeration framework within the open-source Arachne framework, enabling scalable analysis of previously intractable large-scale property graphs.
\end{enumerate}

Comprehensive evaluation on synthetic and real-world graphs demonstrates HiPerMotif's scalability across diverse domains, from computational neuroscience to social network analysis. The rest of this paper is organized as follows: Section \ref{sec:background} reviews related work, Section \ref{sec:hipermotif} presents HiPerMotif algorithms and Arachne implementation, Section \ref{sec:experiments} evaluates performance on synthetic and real-world graphs, and Section \ref{sec:conclusion} concludes with future directions.

%
\IEEEpeerreviewmaketitle

\section{Background and Preliminaries} \label{sec:background}
\subsection{Subgraph Isomorphism and Monomorphism}
Subgraph isomorphism and monomorphism are fundamental problems in graph theory, with applications in neuroscience, systems biology, social network analysis, and cybersecurity~\cite{fan2012graph,quamar2016nscale,akoglu2015graph}. Given graphs \( G_1 = (V_1, E_1) \) and \( G_2 = (V_2, E_2) \) with vertex labels \(\alpha_i: V_i \rightarrow L_V\) and edge labels \(\beta_i: E_i \rightarrow L_E\), subgraph isomorphism seeks an injective function \( f: V_1 \rightarrow V_2 \) such that:
\begin{alignat}{2}
\forall u, v \in V_1 & :\ (u, v) \in E_1 &&\Leftrightarrow (f(u), f(v)) \in E_2, \\
\forall v \in V_1 & :\ \alpha_1(v) &&= \alpha_2(f(v)), \\
\forall (u, v) \in E_1 & :\ \beta_1(u, v) &&= \beta_2(f(u), f(v)).
\end{alignat}

Subgraph monomorphism relaxes the edge constraint to \( (u, v) \in E_1 \Rightarrow (f(u), f(v)) \in E_2 \). Both problems are NP-complete, with worst-case complexity \( O(|V_2|^{|V_1|}) \)~\cite{cook2023complexity}. Property graphs~\cite{angles2018property}, supporting vertex and edge attributes, increase computational demands.

\subsection{Tree-Search and Parallel Algorithms}
Tree-search algorithms dominate exact subgraph matching, using recursive backtracking to explore partial vertex mappings~\cite{ullmann1976algorithm}. VF2 introduced frontier sets for pruning infeasible branches~\cite{cordella2004sub}, while VF3 improved data structures for dense graphs~\cite{carletti2017challenging}. The Glasgow Subgraph Solver employs constraint programming and conflict-directed backjumping~\cite{mccreesh2020glasgow}, and LAD uses arc consistency preprocessing~\cite{solnon2010alldifferent}. Recent methods like DP-iso~\cite{han2019efficient} and VEQ~\cite{kim2021versatile} optimize backtracking with failing sets and dynamic equivalence classes, respectively.

Variable ordering significantly impacts tree-search efficiency~\cite{zhang2024comprehensive}. Static strategies, like RI’s high-degree vertex selection~\cite{bonnici2013ri}, precompute sequences, while dynamic methods, like DP-iso, adapt during search~\cite{han2019efficient}. Techniques such as CFL’s core-forest-leaf decomposition enhance pruning by delaying leaf matching~\cite{bi2016efficient}.

Parallel implementations address scalability on multi-core architectures. VF3P uses state cloning for near-linear scaling up to 16 cores~\cite{carletti2019parallel}, while SLF’s “Ask for Sharing” and “Low-depth Priority Sharing” reduce synchronization overhead~\cite{liang2023slf}. Glasgow integrates parallel backjumping with low-overhead synchronization~\cite{mccreesh2020glasgow}. Dense graphs benefit more from parallelization due to larger search trees, but synchronization limits scalability beyond moderate core counts~\cite{carletti2019parallel}.

Index-based methods (e.g., FG-index~\cite{cheng2007fg}) preprocess graphs for rapid filtering but incur high memory overhead, limiting scalability for large, attribute-rich graphs~\cite{katsarou2017hybrid}. GPU-accelerated approaches like GSI~\cite{zeng2020gsi} leverage parallelism but struggle with irregular memory access~\cite{tran2015fast}.

\subsection{VF2-PS: Our Enhanced Baseline}
Building upon VF2's foundation, our previous VF2-PS algorithm \cite{dindoost2024vf2} introduced parallel execution capabilities while maintaining comprehensive attribute support for both vertex and edge labels. VF2-PS addresses scalability limitations of sequential VF2 through thread-safe state cloning and parallel-safe data structures that minimize synchronization overhead while preserving pruning effectiveness.
The algorithm maintains VF2's sophisticated state-space representation where each state consists of the current partial mapping $M$ and frontier sets $(T_{\text{in}}^1, T_{\text{out}}^1, T_{\text{in}}^2, T_{\text{out}}^2)$ representing the boundary between mapped and unmapped vertices. Parallel execution is achieved through independent exploration of different search branches, with careful state management ensuring thread safety without compromising algorithmic correctness.
While VF2-PS demonstrates significant improvements over sequential VF2, fundamental limitations persist that motivate the innovations in HiPerMotif. The algorithm lacks explicit domain tracking and constraint propagation mechanisms found in advanced systems like Glasgow, resulting in redundant exploration of infeasible branches.

\section{HIPERMOTIF ALGORITHM}
\label{sec:hipermotif}
HiPerMotif addresses subgraph isomorphism through a hybrid approach that fundamentally changes the search initialization strategy. Instead of traditional vertex-by-vertex exploration from empty states, HiPerMotif first systematically identifies all valid mappings for the first edge in the reordered pattern graph, validates these candidates efficiently, and injects them as partial states at depth 2 into traditional tree-search. This eliminates expensive early search tree exploration while enabling natural parallelization over edge candidates.

Figure~\ref{fig:workflow} illustrates the complete workflow: structural reordering prioritizes high-degree vertices, systematic edge enumeration identifies all possible first edge mappings, efficient validators prune infeasible candidates, and state injection initializes traditional matching at depth 2.

\subsection{Structural Reordering}
We reorder the pattern graph \( G_1 \) to prioritize high-degree vertices, enhancing pruning in subsequent matching. A permutation \(\pi: V_1 \to \{0, 1, \dots, |V_1|-1\}\) reindexes vertices based on a ranking function \(\sigma(v) = (\text{totaldeg}(v), \text{outdeg}(v))\), where \(\text{totaldeg}(v)\) is the sum of in- and out-degrees, and ties are broken by \(\text{outdeg}(v)\). Algorithm~\ref{alg:structural-reordering} iteratively selects the highest-ranked vertex \( v^* \), swaps it to the first position, and updates edge arrays \(\texttt{src}\) and \(\texttt{dst}\). Subsequent vertices are chosen from unplaced out-neighbors of the last placed vertex, maximizing \(\sigma(w)\), or from remaining vertices if no out-neighbors exist. This preserves a path-like structure, improving matching efficiency.

\begin{algorithm}[htbp]
\caption{Structural Reordering}
\label{alg:structural-reordering}
\begin{algorithmic}[1]
    \Procedure{\texttt{StructuralReorder}}{\texttt{src}, \texttt{dst}}
        \State Compute \(\text{deg}(v)\) for all \( v \in V \)
        \State \(\mathcal{R} \gets \varnothing\)
        \State \( v^* \gets \underset{v \in V}{\arg\max} \sigma(v) \)
        \State \Call{Swap}{$v^*, \pi(1)$}
        \State \(\mathcal{R} \gets \mathcal{R} \cup \{ v^* \}\)
        \While{\( |\mathcal{R}| < |V| \)}
            \State \( u \gets \pi(|\mathcal{R}|) \)
            \State Update \(\text{indeg}, \text{outdeg}, \text{totaldeg}\)
            \State \( N^+(u) \gets \{ w \notin \mathcal{R} \mid (u \to w) \in E \} \)
            \If{\( N^+(u) \neq \varnothing \)}
                \State \( w^* \gets \underset{w \in N^+(u)}{\arg\max} \sigma(w) \)
            \Else
                \State \( w^* \gets \underset{w \in V \setminus \mathcal{R}}{\arg\max} \sigma(w) \)
            \EndIf
            \State \Call{Swap}{$w^*, \pi(|\mathcal{R}| + 1)$}
            \State \(\mathcal{R} \gets \mathcal{R} \cup \{ w^* \}\)
        \EndWhile
        \State \Return \( (\texttt{src}, \texttt{dst}, \pi) \)
    \EndProcedure
\end{algorithmic}
\end{algorithm}

The algorithm produces a valid permutation, preserving graph structure, with time complexity \( O(|V_1|^2) \) and space complexity \( O(|V_1| +  |E_1|) \). This quadratic preprocessing cost is negligible since \( |V_1| \ll |V_2| \) in practical subgraph matching scenarios, where pattern graphs typically represent small motifs while target graphs may contain millions of vertices.

\subsection{Matching-optimal Viable Edge Selection and Validation}
After structural reordering, the first edge \((0, 1)\) connecting the highest-degree vertices becomes the Matching-optimal Viable Edge (MVE) - the optimal starting point for edge-centric initialization due to its high connectivity and pruning potential. MVE selection systematically identifies all edges in the target graph \( G_2 = (V_2, E_2) \) that can potentially map to this first edge \((0, 1)\) in the reordered \( G_1 \):
\[
E_2^* = \{ (u, v) \in E_2 \mid \texttt{checkAttributes}((u, v), (0, 1)) \},
\]
where \(\texttt{checkAttributes}\) verifies vertex and edge attribute compatibility. Rather than exploring all possible vertex assignments sequentially, this approach directly targets promising edge mappings. Two efficient validation mechanisms ensure early pruning:

\begin{proposition}[Correctness of Structural Reordering]
Algorithm~\ref{alg:structural-reordering} produces a valid permutation \(\pi: V_1 \to \{0, 1, \dots, |V_1|-1\}\) that preserves all vertex and edge information from the original graph \( G_1 \).
\end{proposition}
\begin{proof}
Algorithm~\ref{alg:structural-reordering} initializes set \(\mathcal{R}\) as empty (line 3). It selects the highest-ranked vertex \( v^* \) by \(\sigma(v) = (\text{totaldeg}(v), \text{outdeg}(v))\) (line 4), swapping it to position \(\pi(1)\) (line 5). In each iteration of the while loop (lines 7--18), a unique vertex \( w^* \) is chosen from unplaced out-neighbors or remaining vertices (lines 12--14), swapped to \(\pi(|\mathcal{R}| + 1)\) (line 16), and added to \(\mathcal{R}\) (line 17). The \(\texttt{Swap}\) operations (lines 5, 16) update \(\texttt{src}\) and \(\texttt{dst}\), preserving edge connectivity. Since \(\mathcal{R}\) grows by one vertex per iteration until \( |\mathcal{R}| = |V_1| \), each vertex receives a unique index in \(\{0, 1, \dots, n-1\}\), ensuring \(\pi\) is a valid permutation maintaining \( G_1 \)’s structure.
\end{proof}

\begin{theorem}[Complexity of Structural Reordering]
Algorithm~\ref{alg:structural-reordering} has time complexity \( O(|V_1|^2) \) and space complexity \( O(|V_1| + |E_1|) \), where \( n = |V_1| \).
\end{theorem}
\begin{proof}
Computing degree metrics (line 2) requires scanning all edges, taking \( O(|V_1| + |E_1|) \) time. The main loop (lines 7--18) executes \( |V_1|-1 \) iterations. In each iteration, identifying unplaced out-neighbors (line 10) takes \( O(\text{deg}(u)) \), and finding the maximum-ranked vertex (lines 12 or 14) requires \( O(|V_1|) \) time in the worst case, as \(\sigma(w)\) is computed for all unplaced vertices. The \(\texttt{Swap}\) operation (line 16) updates all edges involving \( w^* \), taking \( O(|E_1|) \) time across all iterations, as each edge is updated at most twice. Thus, the total time complexity is \( O(|V_1| \cdot |V_1| + |E_1|) = O(|V_1|^2 + |E_1|) \), dominated by \( O(n^2) \) for dense graphs. Space complexity includes the graph (\( O(|V_1| + |E_1|) \)), \(\mathcal{R}\) (\( O(|V_1|) \)), and auxiliary arrays (\( O(|V_1|) \)), yielding \( O(|V_1|+ |E_1|) \).
\end{proof}
\textbf{Vertex Validator} (Algorithm~\ref{alg:vertex-validator}): Identifies candidate vertices \( v \in V_2 \) for vertex 0 in \( G_1 \), checking attribute matches and degree constraints (\( T_{\text{in}}^v \geq T_{\text{in}}^0 \), \( T_{\text{out}}^v \geq T_{\text{out}}^0 \)), where \( T_{\text{in/out}}^0 = |N_{G_1}^{\text{in/out}}(0)| \). It outputs a boolean array \(\texttt{vertexFlag}\), pruning infeasible vertices.

\begin{algorithm}[htbp]
\caption{Vertex Validator}
\label{alg:vertex-validator}
\begin{algorithmic}[1]
    \Procedure{\texttt{VV}}{$G_1, G_2$}
        \State $\texttt{vertexFlag} = [a_1, \dots, a_n]$
        \State $T_{\text{in}}^0 \gets |N_{G_1}^{\text{in}}(0)|$, $T_{\text{out}}^0 \gets |N_{G_1}^{\text{out}}(0)|$
        \ForAll{$v \in V(G_2)$}
            \State $T_{\text{in}}^v \gets |N_{G_2}^{\text{in}}(v)|$, $T_{\text{out}}^v \gets |N_{G_2}^{\text{out}}(v)|$
            \If{$\texttt{checkAttributes}(v, 0) \wedge T_{\text{in}}^v \geq T_{\text{in}}^0 \wedge T_{\text{out}}^v \geq T_{\text{out}}^0$}
                \State $\texttt{vertexFlag}[v] \gets \texttt{true}$
            \EndIf
        \EndFor
        \State \Return $\texttt{vertexFlag}$
    \EndProcedure
\end{algorithmic}
\end{algorithm}

\textbf{Edge Validator} (Algorithm~\ref{alg:edge-validator}): Verifies if edge \((u, v) \in E_2\) maps to the MVE \((0, 1) \in E_1\), checking attributes, degree thresholds for vertex 1, bidirectionality, and neighbor overlap (\(|N_u \cap N_v| \geq |N_0 \cap N_1|\)). It generates and updates state \( s \), setting \(\texttt{core}[0, 1] = (u, v)\) and depth of the mapping to 2.

\begin{algorithm}[htbp]
\caption{Edge Validator}
\label{alg:edge-validator}
\begin{algorithmic}[1]
    \Procedure{EV}{$u, v, s$}
        \State $T_{\text{in/out}}^{u,v} \gets N_{G_2}^{\pm}(u, v)$; \quad $T_{\text{in/out}}^{0,1} \gets N_{G_1}^{\pm}(0, 1)$
        \If{$\neg \texttt{match}(v, 1)$} \Return \texttt{false} \EndIf
        \State $e_1 \gets \texttt{getEdgeId}(u, v)$; \quad $e_1^r \gets \texttt{getEdgeId}(v, u)$
        \State $e_2 \gets \texttt{getEdgeId}(0, 1)$; \quad $e_2^r \gets \texttt{getEdgeId}(1, 0)$
       \If{$\neg \texttt{match}(e_1, e_2) \vee (e_2^r \neq -1 \land e_1^r = -1)$}
    \State \Return \texttt{false}
\EndIf

\If{$e_1^r \neq -1 \land e_2^r \neq -1 \land \neg \texttt{checkAttributes}(e_1^r, e_2^r)$}
    \State \Return \texttt{false}
\EndIf

        \If{$|T_{\text{in}}^v| < |T_{\text{in}}^1| \vee |T_{\text{out}}^v| < |T_{\text{out}}^1|$} \Return \texttt{false} \EndIf
        \State $N_{u,v} \gets T_{\text{in}}^u \cup T_{\text{out}}^u \cup T_{\text{in}}^v \cup T_{\text{out}}^v$
        \State $N_{0,1} \gets T_{\text{in}}^0 \cup T_{\text{out}}^0 \cup T_{\text{in}}^1 \cup T_{\text{out}}^1$
        \If{$|N_u \cap N_v| < |N_0 \cap N_1|$} \Return \texttt{false} \EndIf
        \State $s.T_{\text{in/out}}^{G_2} \gets T_{\text{in/out}}^u \cup T_{\text{in/out}}^v \setminus \{u, v\}$
        \State $s.T_{\text{in/out}}^{G_1} \gets T_{\text{in/out}}^0 \cup T_{\text{in/out}}^1 \setminus \{0, 1\}$
        \State $s.\text{depth} \gets s.\text{depth} + 2$
        \State $s.\text{core}[0] \gets u$; \quad $s.\text{core}[1] \gets v$
        \State \Return \texttt{true}
    \EndProcedure
\end{algorithmic}
\end{algorithm}

\begin{proposition}[Correctness of Vertex Validator]
If \( v \in V_2 \) fails the conditions in Algorithm~\ref{alg:vertex-validator}, then \( v \) cannot serve as a valid mapping for vertex 0 in any subgraph match from \( G_1 \) to \( G_2 \).
\end{proposition}
\begin{proof}
Algorithm~\ref{alg:vertex-validator} checks if \( v \in V_2 \) satisfies \(\texttt{checkAttributes}(v, 0)\), ensuring attribute matches, and degree constraints \( T_{\text{in}}^v \geq T_{\text{in}}^0 \), \( T_{\text{out}}^v \geq T_{\text{out}}^0 \), where \( T_{\text{in/out}}^0 = |N_{G_1}^{\text{in/out}}(0)| \) and \( T_{\text{in/out}}^v = |N_{G_2}^{\text{in/out}}(v)| \) (lines 3--9). If \(\texttt{checkAttributes}(v, 0)\) is false, \( v \)'s attributes violate the label constraints for vertex 0. If \( T_{\text{in}}^v < T_{\text{in}}^0 \) or \( T_{\text{out}}^v < T_{\text{out}}^0 \), \( v \) lacks sufficient in- or out-neighbors to match 0's adjacency structure in \( G_1 \). Thus, \( v \) cannot be part of any valid embedding \(\phi\) with \(\phi(0) = v\), ensuring correctness of pruning.
\end{proof}

\begin{proposition}[Correctness of Edge Validator]\label{prop:edge_validator}
If Algorithm~\ref{alg:edge-validator} returns \(\texttt{true}\) for edge \((u, v) \in E_2\), then \(\{0 \mapsto u, 1 \mapsto v\}\) forms a valid partial mapping extendable to a subgraph matching of \( G_1 \) into \( G_2 \). If it returns \(\texttt{false}\), \((u, v)\) cannot map to \((0, 1) \in E_1\).
\end{proposition}
\begin{proof}
Algorithm~\ref{alg:edge-validator} verifies that \((u, v) \in E_2\) matches \((0, 1) \in E_1\) (line 3); (2) edge attribute matches via \(\texttt{match}(e_1, e_2)\) (line 6); (3) bidirectionality, ensuring \((v, u) \in E_2\) exists if \((1, 0) \in E_1\) (lines 7--8); (4) degree thresholds \( |T_{\text{in}}^v| \geq |T_{\text{in}}^1| \), \( |T_{\text{out}}^v| \geq |T_{\text{out}}^1| \) (line 8); and (5) neighbor overlap \( |N_u \cap N_v| \geq |N_0 \cap N_1| \) (line 11). If any check fails, \((u, v)\) violates structural or label constraints, precluding a valid mapping. If all checks pass, the partial mapping \(\{0 \mapsto u, 1 \mapsto v\}\) satisfies adjacency and label constraints, extendable by VF2-PS, ensuring correctness. 
\end{proof}
\subsection{Hybrid Algorithm with State Injection}
Algorithm~\ref{alg:si} implements the hybrid approach by injecting pre-validated MVE mappings into traditional tree-search. For each potentially valid edge mapping \((u, v) \rightarrow (0, 1)\), the algorithm creates a partial state \( s \) with the first two vertices already mapped and injects it at depth 2 into VF2-PS. This bypasses the expensive exploration of depths 0-1 where pruning opportunities are limited, while preserving the efficiency of established tree-search methods for deeper exploration.

\begin{algorithm}[htbp]
\caption{HiPerMotif Algorithm}
\label{alg:si}
\begin{algorithmic}[1]
    \Procedure{HiPerMotif}{$G_1, G_2$}
        \State $M \gets \text{new list(int)}$
        \ForAll{$e \in E_2$} \Comment{Parallel over edges}
            \State $u \gets \text{src}(e)$, $v \gets \text{dst}(e)$
            \If{$\texttt{vertexFlag}[u] \land u \neq v$}
                \State $s \gets \text{new State}(|V_2|, |V_1|)$
                \If{$\texttt{EV}(u, v, s)$}
                    \State $M_{\text{new}} \gets \texttt{VF2-PS}(s, 2)$
                    \State $M \gets M \cup M_{\text{new}}$
                \EndIf
            \EndIf
        \EndFor
        \State \Return $M$
    \EndProcedure
\end{algorithmic}
\end{algorithm}

The key innovation lies in the state injection mechanism: instead of VF2-PS starting from an empty mapping and sequentially assigning vertex 0, then vertex 1, HiPerMotif pre-validates all possible \((u, v) \rightarrow (0, 1)\) MVE mappings and directly initializes VF2-PS with \(\texttt{core}[0] = u, \texttt{core}[1] = v\) at depth 2. This hybrid strategy enables natural parallelization over \(|E_2|\) edge candidates while maintaining the proven correctness and completeness of VF2-PS for the remaining vertex assignments.

\begin{figure*}[t]
    \centering
    \includegraphics[width=\textwidth]{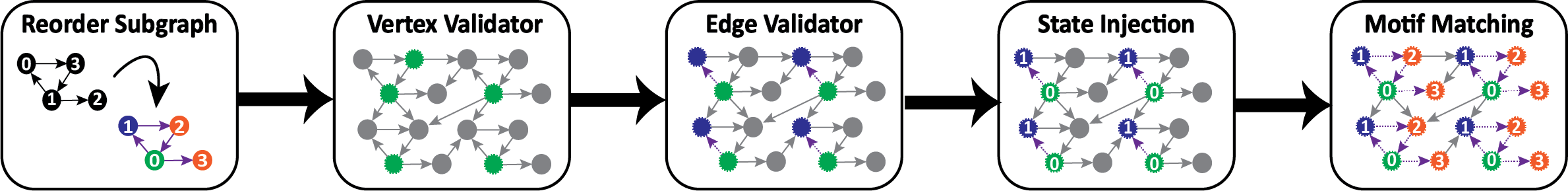}
    \caption{HiPerMotif's workflow: structural reordering, MVE selection, vertex/edge validation, and state injection.}
    \label{fig:workflow}
\end{figure*}

\begin{theorem}[Correctness of Injecting States]
Every mapping \(\phi \in M\) returned by Algorithm~\ref{alg:si} is a correct embedding of \( G_1 \) in \( G_2 \), satisfying adjacency and label constraints.
\end{theorem}
\begin{proof}
Algorithm~\ref{alg:si} processes edges \((u, v) \in E_2\) where \(\texttt{vertexFlag}[u] = \text{true}\) (line 5). If \(\texttt{EV}(u, v, s)\) returns \(\text{true}\) (line 7), Proposition~\ref{prop:edge_validator} ensures \(\{0 \mapsto u, 1 \mapsto v\}\) is a valid partial mapping. The state \( s \), with \(\texttt{core}[0, 1] = (u, v)\) and depth 2, is passed to VF2-PS (line 8), which enforces adjacency and label constraints for remaining vertices~\cite{dindoost2024vf2}. Since VF2-PS only extends valid partial mappings, all \(\phi \in M\) (line 9) satisfy the subgraph isomorphism constraints, ensuring correctness.
\end{proof}

\begin{theorem}[Completeness of State Injection]
If there exists a valid embedding \(\phi: V_1 \to V_2\) with \(\phi(0) = u\), \(\phi(1) = v\), and \(\texttt{EV}(u, v, s) = \text{true}\), then \(\phi \in M\). Thus, HiPerMotif finds all embeddings that VF2-PS would find for initial edge \((u, v)\) matching \((0, 1)\).
\end{theorem}
\begin{proof}
For a valid embedding \(\phi\), \(\texttt{checkAttributes}(u \to v, 0 \to 1) = \text{true}\), \(\texttt{checkAttributes}(v, 1) = \text{true}\), and degree/neighbor criteria hold (Section III-B), so \(\texttt{EV}(u, v, s) = \text{true}\) (line 7). Algorithm~\ref{alg:si} passes state \( s \) to VF2-PS at depth 2 (line 8), which exhaustively enumerates all extensions of \(\{0 \mapsto u, 1 \mapsto v\}\)~\cite{dindoost2024vf2}. Since VF2-PS is complete, \(\phi\) is included in \( M_{\text{new}} \) and thus \( M \) (line 9). Hence, HiPerMotif captures all valid embeddings for \((u, v)\) matching \((0, 1)\), preserving VF2-PS’s completeness.
\end{proof}

HiPerMotif is correct and complete, finding all valid embeddings as VF2-PS when \((u, v)\) matches the MVE \((0, 1)\). Time complexity is \( O(|E_2| \cdot |V_2|) \) in the worst case, but sparse graphs approach \( O(|E_2|) \). Space complexity is \( O(|V_1| + |V_2|) \) per state, lower than domain-aware methods like LAD~\cite{solnon2010alldifferent}. Reordering reduces candidate vertices, boosting pruning, as verified in experiments (Section~\ref{sec:experiments}).

\subsection{Enhancing the Arachne Framework} \label{sec:arachne}
HiPerMotif is implemented in Arachne~\cite{rodriguez2022arachne,du2021interactive}, an open-source framework for large-scale graph analytics with a Python API and Chapel~\cite{chamberlain2007parallel} base servers\cite{bader2023triangle, rodriguez2023property, rodriguez2022arachne}. Our enhancements optimize Arachne for attribute-rich graphs, supporting arbitrary numbers of vertex and edge properties. Graphs are stored in a double-index (DI) format~\cite{a14080221,rodriguez2023property}, extending compressed-sparse row (CSR) with an edge-to-source array for \( O(1) \) edge access, critical for the edge validator's frequent endpoint queries (Algorithm~\ref{alg:edge-validator}). Attributes are managed in contiguous arrays~\cite{rodriguez2023property}, enabling \( O(1) \) vertex and \( O(\log d) \) edge attribute access (where \( d \) is the maximum out-degree), enhancing the efficiency of \(\texttt{checkAttributes}\) in vertex and edge validators (Algorithms~\ref{alg:vertex-validator},~\ref{alg:edge-validator}). Chapel's concurrency ensures parallel scalability for HiPerMotif's reordering and matching. These improvements enable high-performance motif finding on large graphs, as validated in Section~\ref{sec:experiments}.

\section{Experimental Evaluation} \label{sec:experiments}

We evaluate HiPerMotif against state-of-the-art parallel subgraph matching algorithms: VF2-PS \cite{dindoost2024vf2}, VF3P \cite{carletti2019parallel}, and Glasgow Subgraph Solver \cite{mccreesh2020glasgow} across synthetic and real-world datasets. All experiments report mean execution times from 5 independent runs with 95\% confidence intervals. Experiments were conducted on dual AMD EPYC 7713 processors (128 cores total, 2.0GHz) with 512GB RAM.

\subsection{Datasets}

\subsubsection{Synthetic Graph Generation}

For controlled experimental conditions, we developed several families of synthetic graphs that allow systematic evaluation of algorithm performance across different network topologies. Our synthetic dataset collection includes Erd\H{o}s--R\'enyi random graphs created with carefully selected edge probability parameters to represent different network densities. We generated sparse networks with edge probability $P = 0.0005$, medium density networks with $P = 0.005$, and dense networks with $P = 0.05$, spanning size ranges from 5,000 to 130,000 vertices to evaluate scalability characteristics.

Beyond random graphs, we incorporated small-world networks generated using the Watts--Strogatz model with various rewiring probabilities to capture the clustering and short path length properties common in many real-world networks. We also created scale-free networks using the Barab\'asi--Albert preferential attachment model with different parameter settings to represent networks with power-law degree distributions typical of social and biological systems.

Attributes were assigned using uniform random distribution.

\subsubsection{Real-world Network Collections}

Our evaluation incorporates diverse real-world datasets from multiple domains to validate practical applicability. The neuroscience connectome datasets form a central component of our evaluation, including the FlyWire and Hemibrain dataset.

Beyond neuroscience applications, we included additional domain-specific directed networks: EU email communication network and twitter. All graphs in our evaluation are directed.

\subsection{Pattern Graph Design and Motif Selection}
Our experimental framework employs a comprehensive collection of query patterns designed to evaluate algorithm performance across different structural complexities and application domains. These patterns include both well-established motifs commonly used in network science literature as fundamental building blocks of complex networks \cite{milo2002network, alon2007network}, as well as randomly generated subgraphs from 3 to 20 nodes to provide broader coverage of structural diversity.

\subsubsection{Synthetic Graph Performance}
\begin{figure*}[htbp]
\centering
\includegraphics[width=0.9\textwidth]{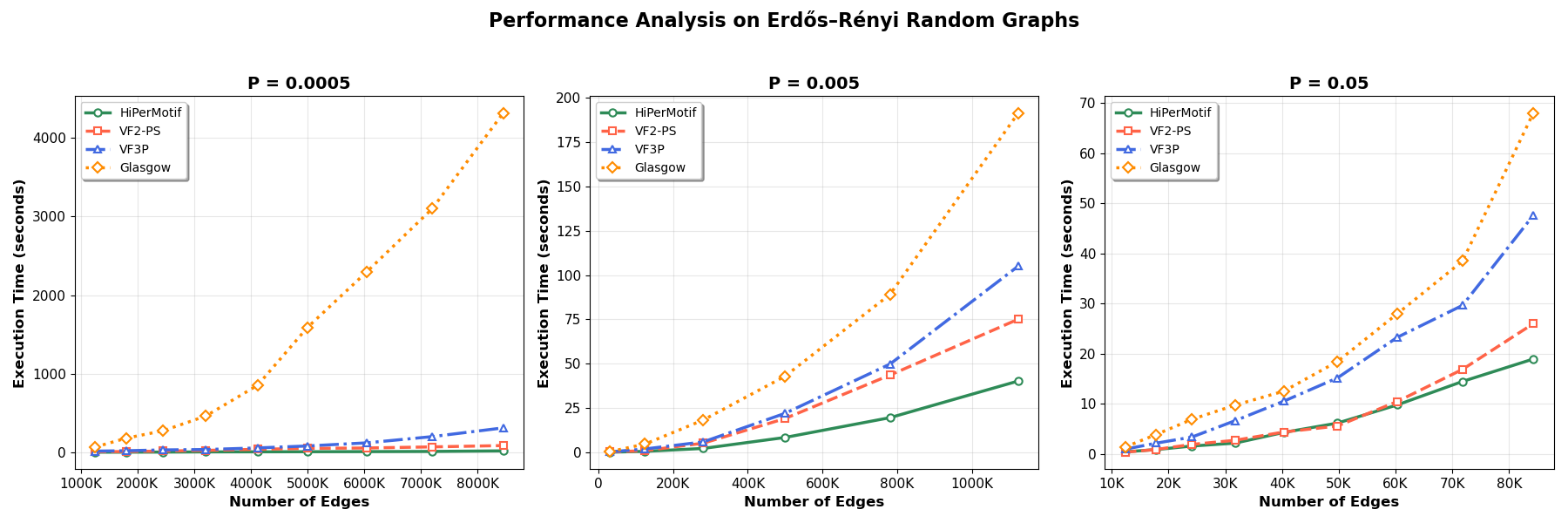}
\caption{Performance comparison on Erdős-Rényi random graphs with varying densities. HiPerMotif demonstrates superior scalability on larger graphs but shows overhead on sparse graphs with few nodes, as it is designed to handle massive and large-scale networks.}
\label{fig:erdos-renyi-performance}
\end{figure*}

The results on Erdős-Rényi random graphs reveal that HiPerMotif exhibits a clear performance advantage as graph size increases, particularly on medium to dense networks. However, when graphs are sparse or contain relatively few nodes, HiPerMotif shows some overhead compared to traditional algorithms. This behavior is expected since HiPerMotif is specifically designed to handle massive and large-scale networks, where its parallel processing capabilities and structural optimizations can be fully utilized. On smaller or very sparse graphs, the overhead of parallel coordination and advanced data structures may outweigh the benefits, making simpler algorithms temporarily more efficient.

\begin{figure}[htbp]
\centering
\includegraphics[width=\columnwidth]{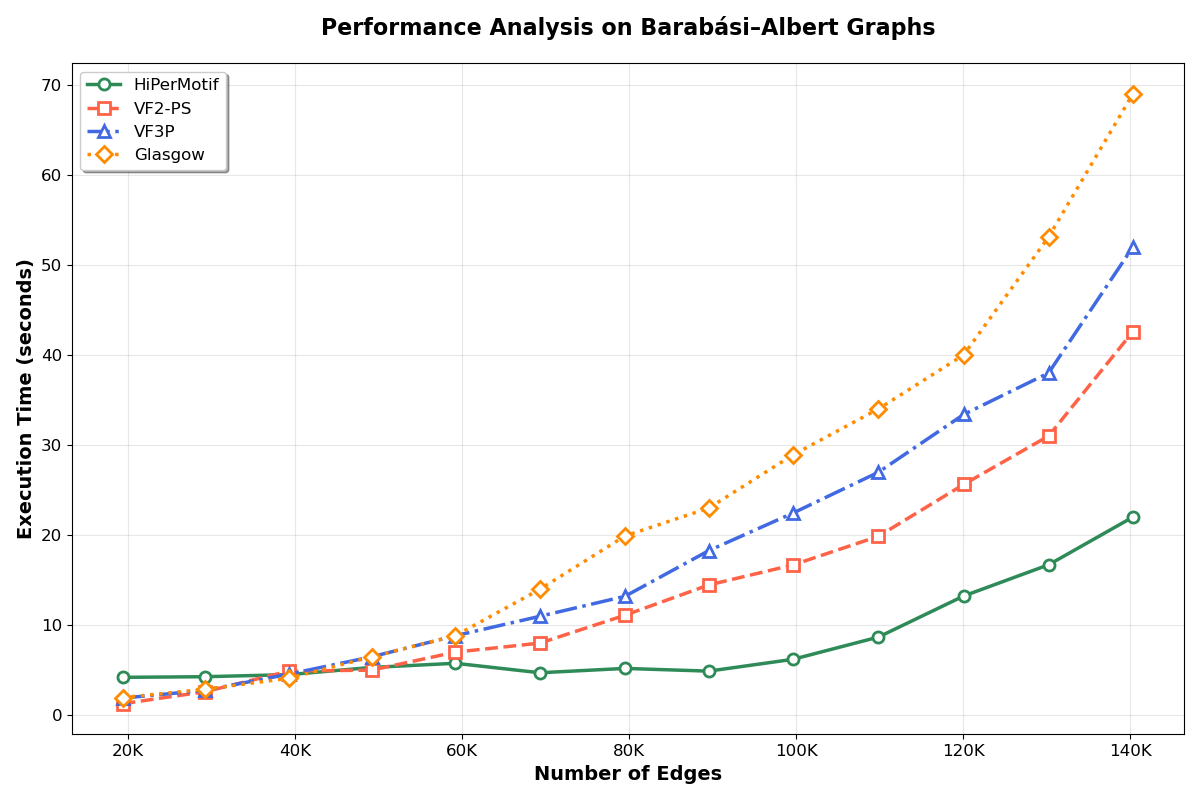}
\caption{Performance comparison on scale-free networks generated with parameters $\alpha = 0.41$, $\beta = 0.54$, $\gamma = 0.05$, $\delta_{in} = 0.2$, $\delta_{out} = 0.2$. The results demonstrate HiPerMotif's effectiveness on networks with power-law degree distributions typical of real-world systems.}
\label{fig:scale-free-performance}
\end{figure}

Scale-free networks, characterized by power-law degree distributions, represent a common topology in many real-world systems including social networks, biological networks, and the internet. We generated these networks using the directed scale-free model with parameters $\alpha = 0.41$, $\beta = 0.54$, $\gamma = 0.05$, $\delta_{in} = 0.2$, and $\delta_{out} = 0.2$. The results show that HiPerMotif maintains its performance advantage on scale-free topologies, effectively handling the heterogeneous degree distribution and hub-dominated structure characteristic of these networks.

\begin{figure}[htbp]
\centering
\includegraphics[width=\columnwidth]{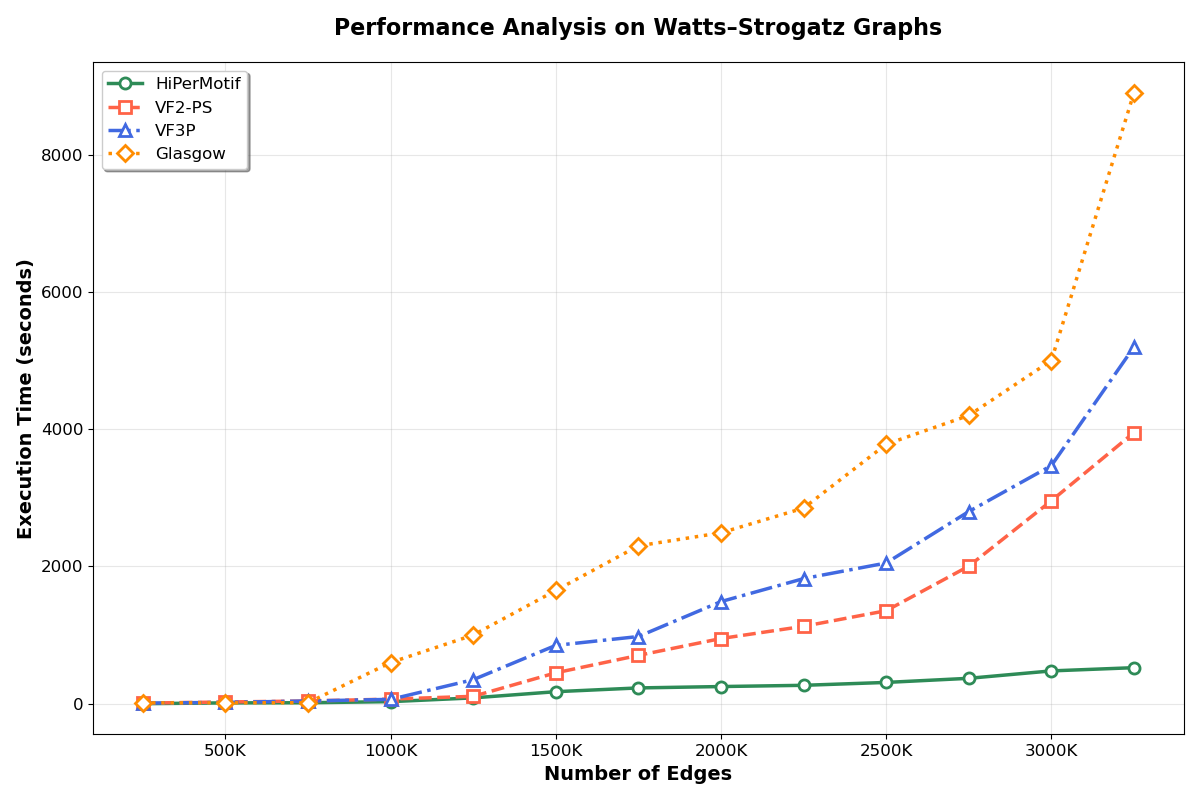}
\caption{Performance comparison on Watts-Strogatz small-world networks generated with parameters $k = 10$, $p = 0.01$, seed = 42. HiPerMotif demonstrates superior performance starting from approximately 250K edges.}
\label{fig:small-world-performance}
\end{figure}

For small-world networks, we employed the Watts-Strogatz model with parameters $k = 10$, $p = 0.01$, and seed = 42 to generate networks that exhibit the characteristic high clustering and short path lengths typical of many real-world systems. The results clearly demonstrate that HiPerMotif maintains competitive performance across all network sizes.

\subsection{Effectiveness of Structural Reordering}

To evaluate the impact of our MVE reordering strategy, we conducted a comprehensive analysis using VF2-PS as the baseline algorithm. We generated 60 random graphs with edge probabilities ranging from 0.0005 to 0.1 and edge counts between 1,000 to 1,500,000. Query patterns were randomly created with sizes from 3 to 20 nodes to provide diverse structural complexity.

\begin{figure}[htbp]
\centering
\includegraphics[width=\columnwidth]{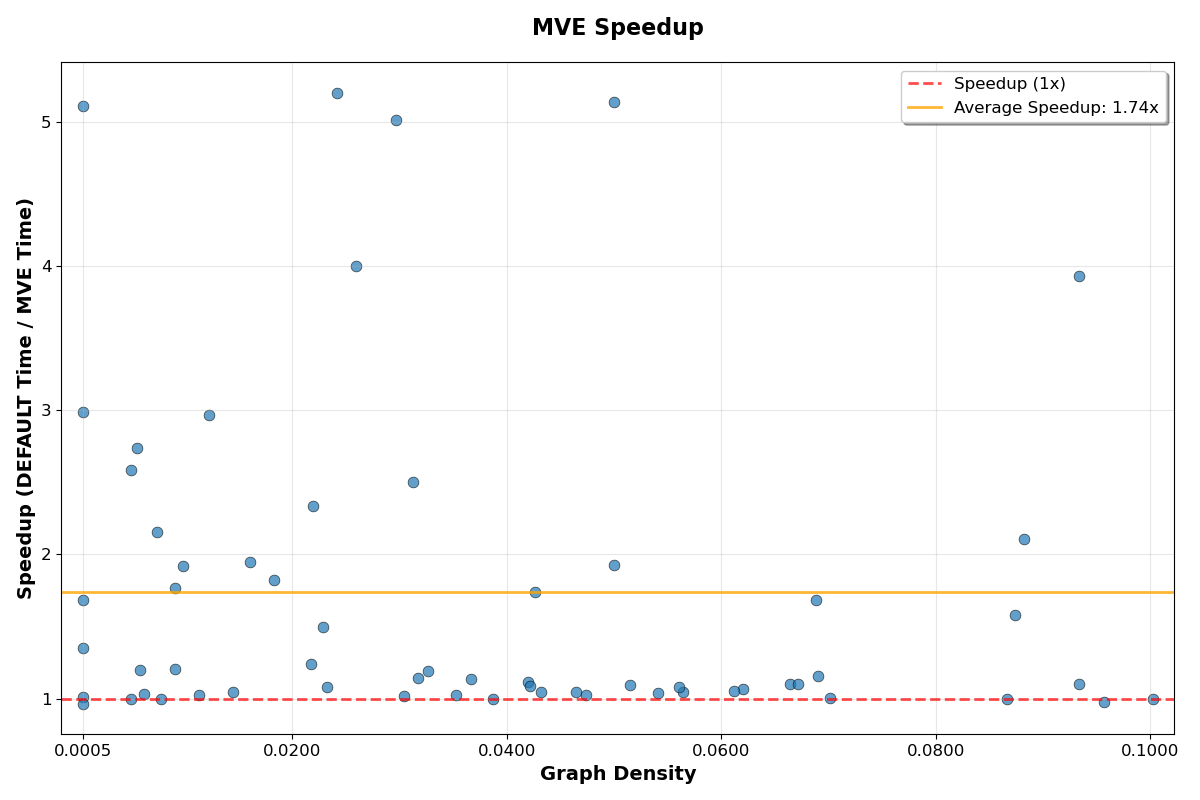}
\caption{Performance improvement achieved by MVE reordering on VF2-PS across diverse graph topologies and query patterns. The reordering strategy shows consistent benefits with peak speedups of 5× and average improvement of 1.74×.}
\label{fig:reordering-effectiveness}
\end{figure}

The results demonstrate the consistent effectiveness of our MVE reordering approach. Across all tested configurations, the reordering strategy introduces no performance degradation while providing substantial benefits in many cases. The peak speedup reached 5× improvement over the baseline VF2-PS, with an average performance gain of 1.74× across all experiments. This confirms that our heuristic structural reordering is beneficial. 

\subsection{Neuroscience Connectome Datasets}

We evaluated algorithm performance on two major neuroscience connectomes that represent different scales of brain reconstruction data.

\textbf{Hemibrain Dataset:} This dataset\cite{scheffer2020connectome} contains 21,739 vertices and 3,550,403 edges representing a partial fruit fly brain reconstruction, providing a realistic test case with complex topological patterns characteristic of biological neural networks.

\begin{figure}[htbp]
\centering
\includegraphics[width=\columnwidth]{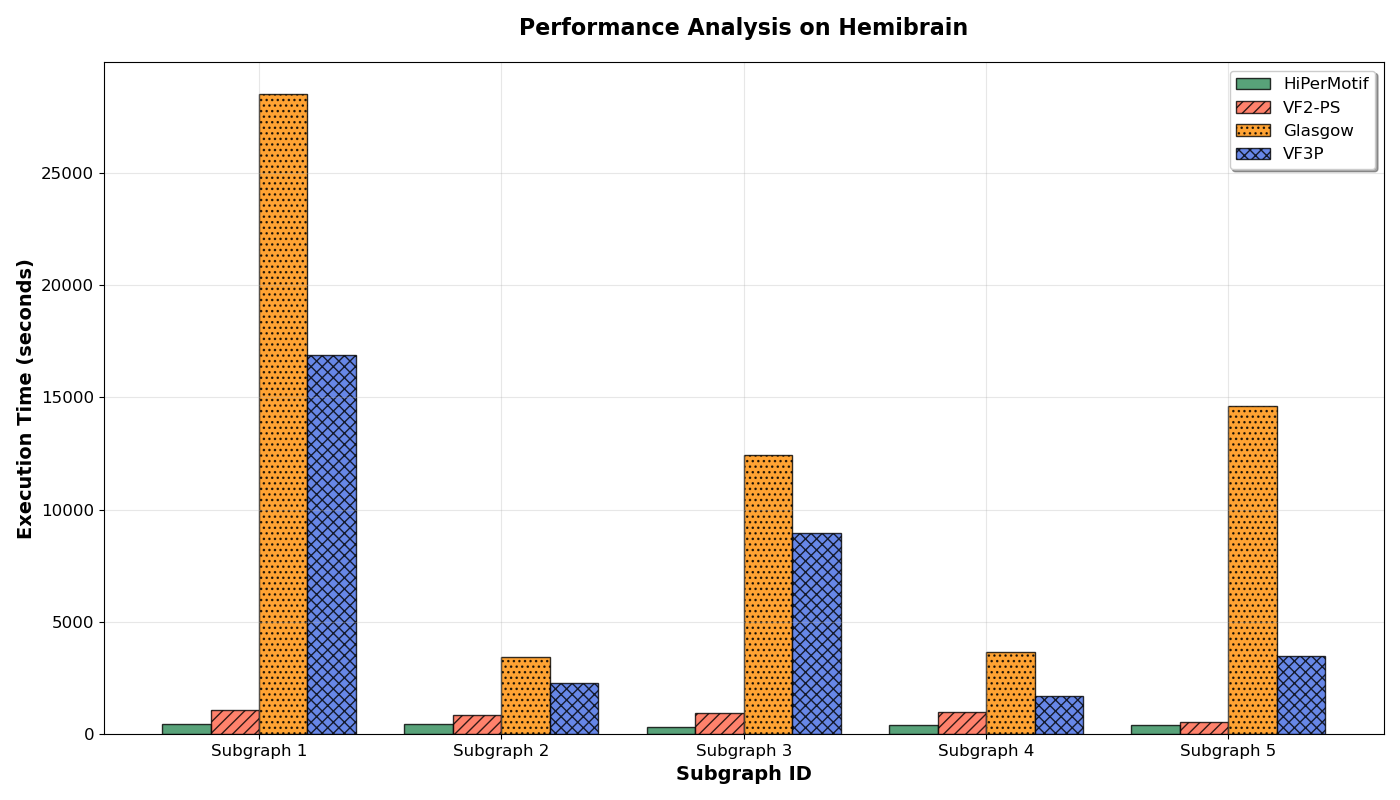}
\caption{Performance comparison on Hemibrain connectome dataset across five randomly generated subgraphs.}
\label{fig:hemibrain-performance}
\end{figure}

\textbf{FlyWire Dataset:} The FlyWire dataset\cite{dorkenwald2024neuronal, schlegel2024whole,zheng2018complete,buhmann2021automatic} represents a complete fruit fly brain connectome with 139,255 vertices and 2,700,513 edges, offering evaluation on a larger-scale neuroscience network.

\begin{figure}[htbp]
\centering
\includegraphics[width=\columnwidth]{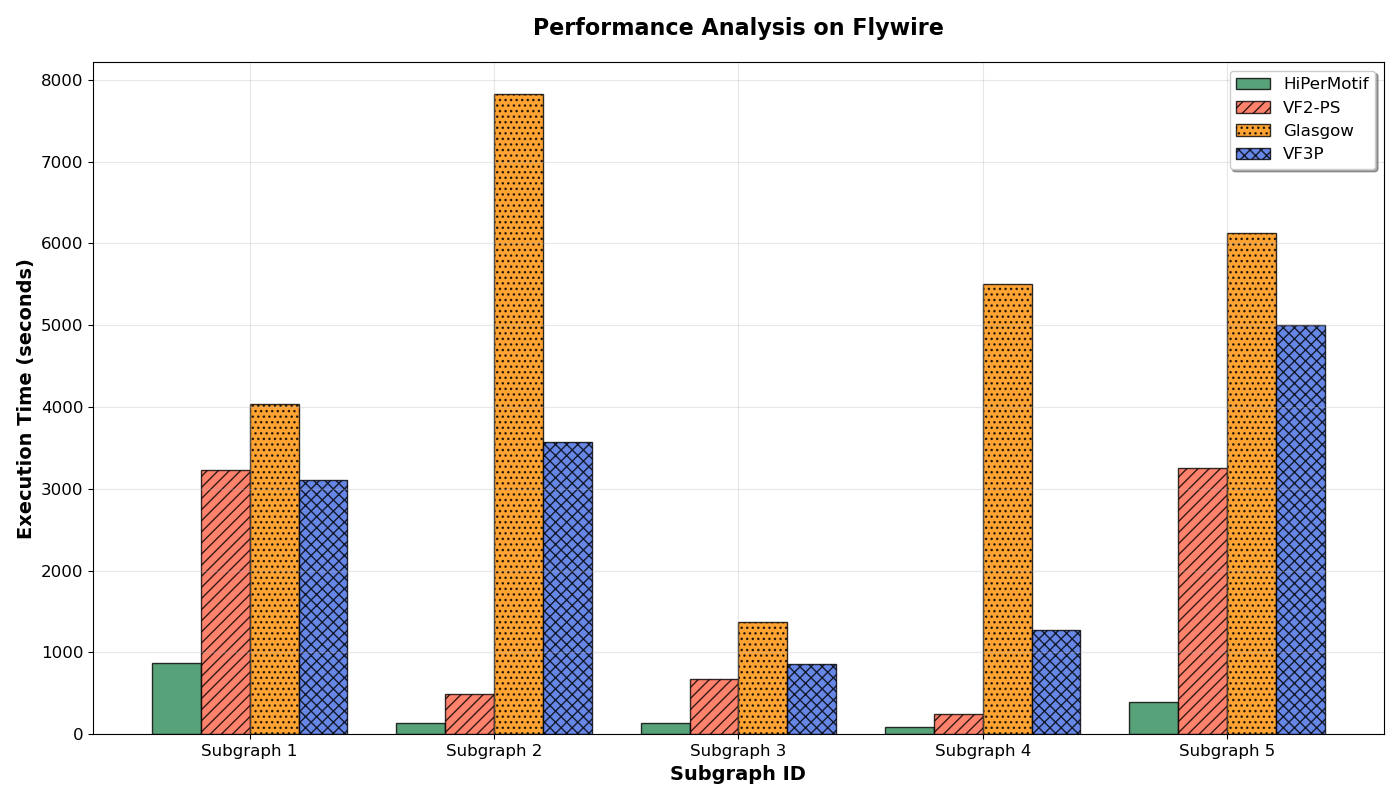}
\caption{Performance comparison on FlyWire connectome dataset across five randomly generated subgraphs.}
\label{fig:flywire-performance}
\end{figure}

The results on both neuroscience datasets demonstrate HiPerMotif's better performance on real-world network structures. Across five different randomly generated subgraph patterns on each dataset, our algorithm consistently outperforms all baseline algorithms and achieves peak speedups of 66.47× on FlyWire and 65.71× on Hemibrain compared to the worst-performing baseline algorithms.

\subsection{Communication Network Datasets}

We evaluated our algorithm on the EU email communication network\cite{leskovec2007graph}, a directed network containing 265,214 vertices and 420,045 edges representing email communications within a European research institution. This dataset provides insights into algorithm performance on communication patterns and organizational network structures.

\begin{figure}[htbp]
\centering
\includegraphics[width=\columnwidth]{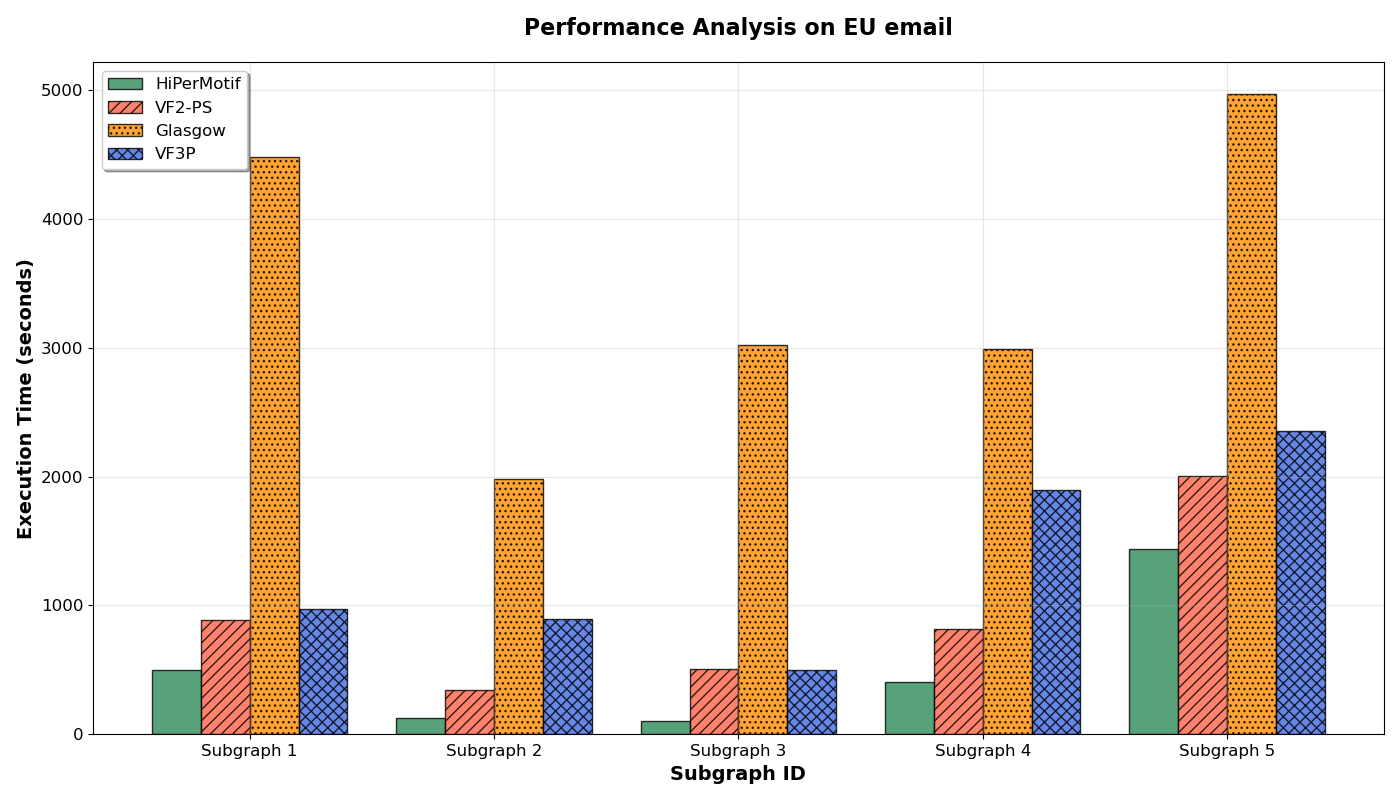}
\caption{Performance comparison on EU email communication network across five randomly generated subgraphs. HiPerMotif demonstrates effective performance on communication network topologies.}
\label{fig:email-communication-performance}
\end{figure}

The results on the EU email network demonstrate effective performance on communication network topologies. Across five different randomly generated subgraph patterns, our algorithm maintains competitive performance compared to baseline algorithms, achieving a peak speedup of 5.92×. This validates the approach's versatility across different network domains beyond biological systems.

\subsection{Social Network Datasets}

We evaluated our algorithm on the Social circles: Twitter dataset\cite{leskovec2012learning}, a directed network containing 81,306 vertices and 1,768,149 edges representing follower relationships and social interactions on Twitter. This dataset allows us to assess performance on social network structures characterized by community formations and influence propagation patterns.

\begin{figure}[htbp]
\centering
\includegraphics[width=\columnwidth]{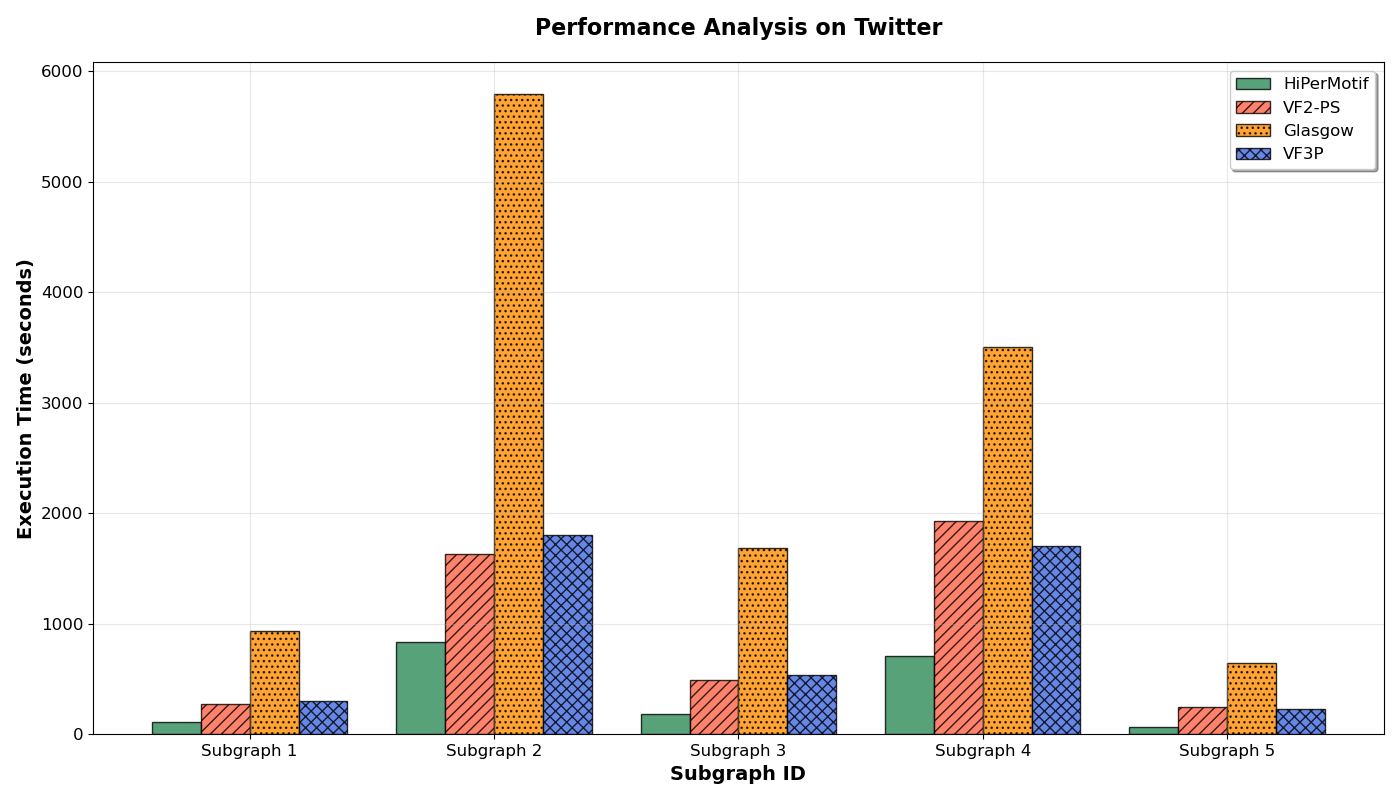}
\caption{Performance comparison on Twitter social circles dataset across five randomly generated subgraphs. The algorithm demonstrates effective performance on social network topologies.}
\label{fig:twitter-social-performance}
\end{figure}

The results on the Twitter social network validate the algorithm's effectiveness across diverse network domains. Across five different randomly generated subgraph patterns, our approach maintains competitive performance compared to baseline algorithms, demonstrating consistent benefits on social network structures with their characteristic clustering and community patterns.

\subsection{Large-Scale Network Analysis: H01 Dataset}

To demonstrate HiPerMotif's capability on truly massive networks, we evaluated performance on the H01 dataset, containing 142,660,662 vertices and 147,071,359 edges representing a cubic millimeter of human cortex. This dataset represents one of the largest connectome reconstructions available and poses significant computational challenges that exceed the capabilities of existing subgraph matching algorithms.

\begin{table}[H]
    \caption{Performance of HiPerMotif on the H01 large-scale dataset. Baseline algorithms could not complete execution due to memory and computational constraints.}
    \centering
    \begin{tabular}{|c|c|}
        \hline
        \textbf{Motif} & \textbf{H01 (seconds)} \\ \hline
        
        \makecell{\vspace{0.05pt}\\\includegraphics[width=1cm]{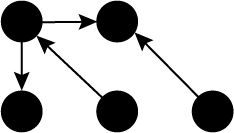}\\\vspace{0.05pt}} & \makecell{571.94} \\ \hline
        \makecell{\vspace{0.05pt}\\\includegraphics[width=1cm]{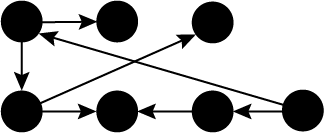}\\\vspace{0.05pt}} & \makecell{1011.62} \\ \hline
        \makecell{\vspace{0.05pt}\\\includegraphics[width=1cm]{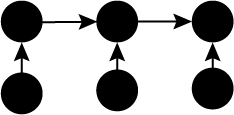}\\\vspace{0.05pt}} & \makecell{21.23} \\ \hline
        \makecell{\vspace{0.05pt}\\\includegraphics[width=1cm]{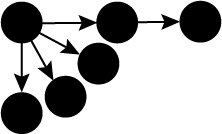}\\\vspace{0.05pt}} & \makecell{363.54} \\ \hline
        \makecell{\vspace{0.05pt}\\\includegraphics[width=1cm]{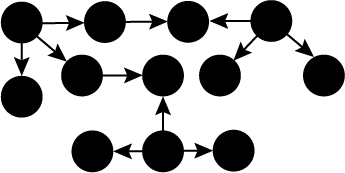}\\\vspace{0.05pt}} & \makecell{1209.82} \\ \hline
    \end{tabular}
    \label{tab:H01_performance}
\end{table}

The H01 results represent a significant achievement in computational neuroscience, as existing baseline algorithms (VF2-PS, VF3P, Glasgow) were unable to complete execution on this massive dataset due to memory limitations and computational complexity. HiPerMotif successfully processed all test motifs, with execution times ranging from 21 seconds to approximately 4 minutes, demonstrating practical feasibility for large-scale connectome analysis. These performance gains are particularly attributable to our Vertex Validator component, which enables rapid pruning of infeasible search paths and significantly reduces the computational complexity on such massive networks. This capability gap highlights the fundamental scalability advantages of our approach and opens new possibilities for analyzing the largest available brain reconstruction datasets.

\subsection{Parallel Speedup Analysis}

We evaluate the parallel performance of HiPerMotif on Erdős-Rényi networks of various sizes. Speedup is defined as $S = \frac{T_{1}}{T_{P}}$, where $T_{1}$ is the execution time of the single-threaded version and $T_{P}$ is the execution time using $P$ threads.

\begin{figure*}[htbp]
    \centering
    \includegraphics[width=0.9\textwidth]{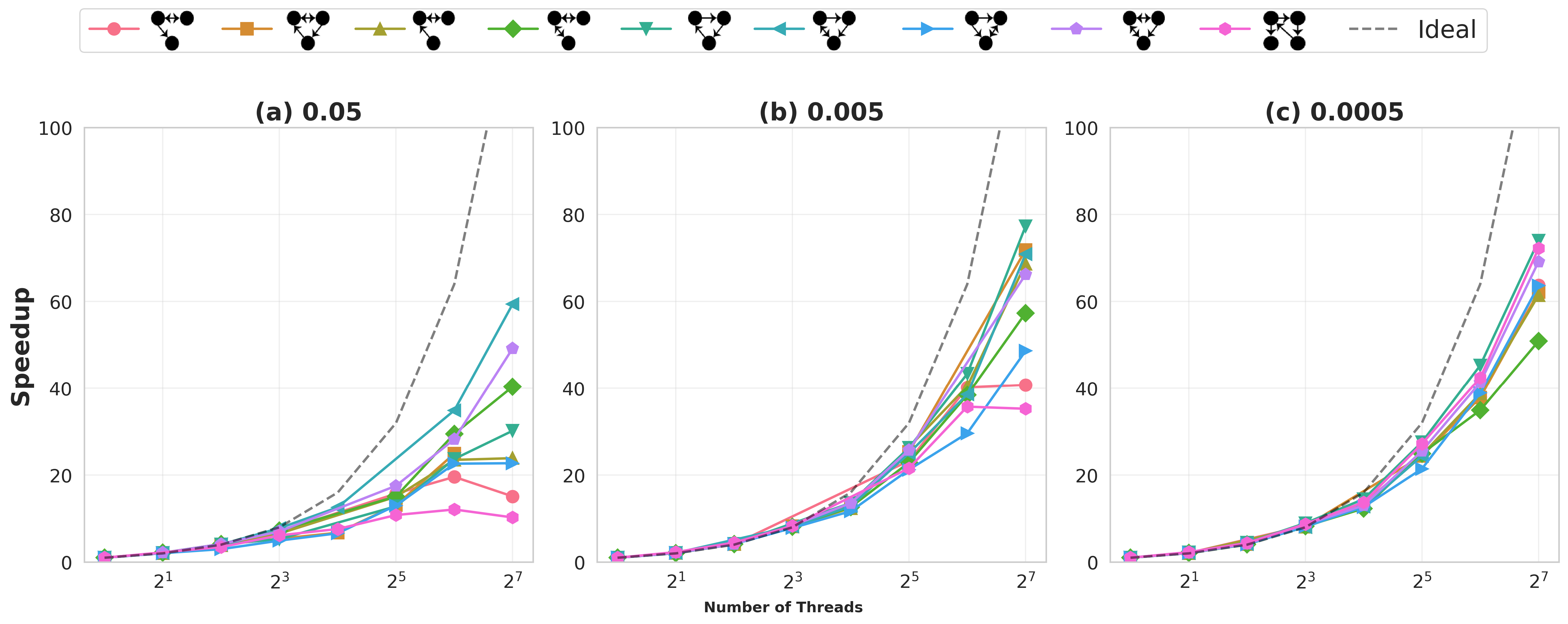}
    \caption{Parallel speedup of HiPerMotif across Erdős-Rényi networks with varying vertex counts and edge probabilities ($p$ = 0.05, 0.005, 0.0005).}
    \label{fig:hipermotif_scalability}
\end{figure*}

The results demonstrate HiPerMotif's parallel scalability, particularly for large-scale graphs where near-linear speedup is frequently achieved. This superior performance stems from several key algorithmic properties. The Vertex Validator rapidly eliminates infeasible states through early pruning, minimizing synchronization overhead and aligning with established parallel search optimization principles \cite{mccreesh2015parallel,carletti2018vf3}. Large-scale graphs naturally generate substantial parallel workloads through numerous concurrent matching states, maximizing processor utilization \cite{amdahl1967validity, gustafson1988reevaluating}. Additionally, the recursive, edge-centric architecture enables efficient parallelization by decomposing partial matches into independent subproblems, minimizing shared data structure contention \cite{mccreesh2015parallel,han2013turboiso}.

For smaller or denser graphs, parallel speedup shows more modest gains, consistent with prior parallel subgraph matching algorithms \cite{mccreesh2015parallel, cheng2009efficient}. This behavior reflects fundamental parallel computing constraints: smaller graphs provide limited concurrent exploration opportunities, while Amdahl's law dictates that sequential components become performance-limiting factors at reduced problem scales \cite{amdahl1967validity}.

The experimental validation confirms HiPerMotif's robust scalability for large-scale graphs with near $10^6$ edges, with efficiency gains maximized in contexts with extensive search spaces.


\section{Conclusion}
\label{sec:conclusion}
This paper presented HiPerMotif, a hybrid algorithm for subgraph isomorphism that addresses fundamental scalability limitations of existing approaches on large-scale, attribute-rich directed graphs. Through edge-centric initialization with state injection, structural reordering, and efficient validation mechanisms, HiPerMotif achieves significant performance improvements on large networks while enabling analysis of previously intractable datasets.

Our experimental evaluation demonstrates that HiPerMotif outperforms state-of-the-art algorithms on medium to large-scale networks, achieving speedups of up to 66× on real-world datasets where baselines complete execution, and 5× improvement through structural reordering alone. Notably, HiPerMotif successfully processes massive networks such as the H01 connectome that cause existing algorithms to fail due to memory constraints.
However, HiPerMotif shows performance overhead on small graphs due to parallel coordination costs and advanced data structure overhead. The algorithm is specifically designed for large-scale networks where these overheads are amortized by substantial parallelization benefits. This trade-off represents a conscious design choice prioritizing scalability over small-graph performance.

The broader impact extends beyond performance improvements, enabling research communities working with large-scale networks to tackle previously computationally prohibitive problems in computational neuroscience, social network analysis, and systems biology.

Future research directions include extending HiPerMotif to dynamic graphs, incorporating machine learning-guided edge selection, and developing adaptive mechanisms to automatically switch between HiPerMotif and traditional algorithms based on graph characteristics.

HiPerMotif is open source and is publicly available on GitHub at \url{https://github.com/Bears-R-Us/arkouda-njit} along with the execution scripts and detailed documentation.

\section*{Acknowledgment}
This research was funded in part by NSF grant numbers CCF-2109988,
OAC-2402560, and CCF-2453324.

\bibliographystyle{IEEEtran}
\bibliography{
    bibliography/prop_graph, 
    bibliography/arachne, 
    bibliography/motif,
    bibliography/sub_iso
}
\end{document}